\documentclass[10pt,twocolumn,twoside]{IEEEtran}

\IEEEoverridecommandlockouts

\usepackage{graphics} % for pdf, bitmapped graphics files
\usepackage{epsfig} % for postscript graphics files
\usepackage{mathtools}
\usepackage{algorithm,algorithmic}

\usepackage{balance}
\usepackage{amsmath,amsfonts,amssymb,amscd,amsthm}
\usepackage{capt-of}
\usepackage{bbm}
\usepackage{color}
\usepackage{cite}
\usepackage{verbatim}

\newtheorem{remark}{\bfseries Remark}
\newtheorem{theorem}{\bfseries Theorem}
\newtheorem{lemma}{\bfseries Lemma}
\newtheorem{assump}{\bfseries Assumption}

\newtheorem{defn}{\bfseries Definition}
\newtheorem{prop}{\bfseries Proposition}

\newenvironment{varalgorithm}[1]
  {\algorithm\renewcommand{\thealgorithm}{#1}}
  {\endalgorithm}

\newenvironment{list4}{
	\begin{list}{$\bullet$}{%
			\setlength{\itemsep}{0.05cm}
			\setlength{\labelsep}{0.2cm}
			\setlength{\labelwidth}{0.3cm}
			\setlength{\parsep}{0in} 
			\setlength{\parskip}{0in}
			\setlength{\topsep}{0in} 
			\setlength{\partopsep}{0in}
			\setlength{\leftmargin}{0.16in}}}
	{\end{list}}

\begin{document}

\title{\LARGE \bf
Finite Time Privacy Preserving Quantized Average Consensus with Transmission Stopping}
%\author{ \parbox{3 in}{\centering Huibert Kwakernaak*
%         \thanks{*Use the $\backslash$thanks command to put 
%information here}\\
%         Faculty of Electrical Engineering, Mathematics and 
%Computer Science\\
%         University of Twente\\
%         7500 AE Enschede, The Netherlands\\
%         {\tt\small h.kwakernaak@autsubmit.com}}
%         \hspace*{ 0.5 in}
%         \parbox{3 in}{ \centering Pradeep Misra**
%         \thanks{**The footnote marks may be inserted manually}\\
%        Department of Electrical Engineering \\
%         Wright State University\\
%         Dayton, OH 45435, USA\\
%         {\tt\small pmisra@cs.wright.edu}}
%}

\author{Apostolos~I.~Rikos, Christoforos~N.~Hadjicostis, and
 Karl~H.~Johansson% <-this % stops a space
    \thanks{Apostolos~I.~Rikos and K.~H.~Johansson are with the Division of Decision and Control Systems, KTH Royal Institute of Technology, SE-100 44 Stockholm, Sweden. They are also affiliated with Digital Futures, SE-100 44 Stockholm, Sweden. E-mails:{\tt \{rikos,kallej\}@kth.se}.}
    \thanks{C.~N.~Hadjicostis is with the Department of Electrical and Computer Engineering, University of Cyprus, 1678 Nicosia, Cyprus: E-mail:{\tt~chadjic@ucy.ac.cy}.}
	\thanks{This work was supported by the Knut and Alice Wallenberg Foundation and the Swedish Research Council.}
%\thanks{This work was supported in part by the Cyprus Research Promotion Foundation (CRPF) Framework Programme for Research, Technological Development and Innovation 2009-2010 (CRPF's FP 2009--2010), co-funded by the Republic of Cyprus and the European Regional Development Fund, under \textit{$\Delta I A K P A T I K E \Sigma / K Y - P O Y / 0713 / 21$}. Any opinions, findings, and conclusions or recommendations expressed in this publication are those of the authors and do not necessarily reflect the views of CRPF.}
}
\maketitle
\thispagestyle{empty}
\pagestyle{empty}

%%%%%%%%%%%%%%%%%%%%%%%%%%%%%%%%%%%%%%%%%%%%%%%%%%%%%%%%%%%%%%%%%%%%%%%%%%%%%%%%
\begin{abstract} 
Due to their flexibility, battery powered or energy-harvesting wireless networks are employed in diverse applications. 
Securing data transmissions between wireless devises is of critical importance in order to avoid privacy-sensitive user data leakage. 
In this paper, we focus on the scenario where some nodes are curious (but not malicious) and try to identify the initial states of one (or multiple) other nodes, while some nodes aim to preserve the privacy of their initial states from the curious nodes. 
We present a privacy preserving finite transmission event-triggered quantized average consensus algorithm. 
Its operation is suitable for battery-powered or energy-harvesting wireless network since it guarantees (i) efficient (quantized) communication, and (ii) transmission ceasing (which allows preservation of available energy).
%Each node in the network (including curious nodes) is allowed to execute the proposed privacy preserving algorithm (in order to preserve the privacy of the initial state it contributes to the average calculation) or a variation of the proposed algorithm which is not privacy preserving.
Furthermore, we present topological conditions under which the proposed algorithm allows nodes to preserve their privacy. 
%(ii) obtain, after a finite number of steps, the exact average of the initial states while processing and transmitting quantized information, and (iii) cease transmissions once convergence has been achieved. 
We conclude with a comparison of our algorithm against other algorithms in the existing literature. 
\end{abstract}


% ===============================================
%
%
% INTRODUCTION
%
%
% ===============================================
\section{INTRODUCTION}\label{intro}

Wireless control networks (WCN) play a major role in important applications due to their deployment flexibility, 
% which enables their use wherever a fixed infrastructure is infeasible 
\cite{2018:Park_Ergen_Fischione_Lu_Johansson}. 
Furthermore, the absence of cables for data communication means that nodes need to rely on (i) battery storage, and/or (ii) energy harvesting techniques for their operation. 
Prolonging the lifetime of a device is a topic that has also received a lot of attention in recent years (see \cite{2019:Knorn_Dey_Ahlen_Quevedo} and references therein).
In this paper, in order to prolong the lifetime of nodes we rely on (i) event-triggered operation, (ii) transmission stopping guarantees, and (iii) quantized processing and communication. 

The flexibility of WCN allows them to 
%be deployed and collect important data in hostile environments for a long time or perform permanent monitoring. 
%This means that they usually 
work unattended in hostile environments with a limited energy budget. 
Security and privacy of WCN is a challenging issue since, during their operation in a potentially hostile environment, they are exposed to a variety of privacy attacks. 
%(which also raises the risk of privacy-sensitive user data leakage) 
%and receive limited hardware support due to significant cost constraints. 
Specifically, distributed coordination algorithms require exchange of collected data between neighboring nodes. 
In many occasions there might be nodes in the network that are curious and aim to extract private and/or sensitive data. 
Revealing the state of a node may be undesirable in case the state is private or contains sensitive information.
Efficient (quantized) communication between nodes is a desirable feature since (i) it is more suitable for the available network resources and (ii) exhibits advantages and applicability to public-key cryptosystems.  
%which operate with integer numbers able to encrypt sensitive data. 
For these reasons, several strategies have been proposed for distributed coordination by achieving quantized average consensus \cite{2016:Chamie_Basar, 2020:Rikos_Mass_Accum}. 

\vspace{.2cm}

\noindent
\textbf{Previous Literature.}
There have been different approaches for dealing with the problem of calculating the quantized average of the initial states with privacy preservation guarantees. 
In \cite{2016:CortesPappas, NOZARI:2017} the authors present works on differential privacy. 
In this work, nodes inject uncorrelated noise into the exchanged messages. 
% and the data associated to a particular node cannot be inferred by a curious node. 
%However, the exact average of the initial states is not obtained due to trade-off between privacy and computational accuracy \cite{NOZARI:2017}. 
The injection of correlated noise at each time step and for a finite period of time was proposed in \cite{2013:Nikolas_Hadj}. 
%This strategy overcomes the trade-off in \cite{NOZARI:2017} and guarantees convergence to the exact average, while each node avoids revealing its initial state or the initial states of other nodes. 
In \cite{Mo-Murray:2017} the nodes asymptotically subtract the initial offset values they added in the computation. 
The problem of calculating the average of the initial states in a privacy-preserving manner is discussed in \cite{Rezazadeh_Kia:2018} for a continuous time weight balanced system.  
% and in \cite{Gao_Wang:2018} over directed graphs. 
%Specifically, the nodes are using random coupling weights between connected nodes in order to embed privacy. 
In \cite{Wang:2019} the average of the initial states is calculated in a privacy preserving manner via a state-decomposition-based approach, whereas \cite{2020:Ridgley_Lynch} discusses the problem under certain topological conditions. 
%Their strategy is hot-pluggable and does not need reinitialization when a node enters or leaves the network.  
Homomorphic encryption \cite{ChrisCDC:2018, hadjicostis2020privacy} is another strategy which guarantees privacy preservation. 
However, it requires the existence of trusted nodes and imposes heavier computational requirements. 
% which may result in an increase in energy consumption. 
In \cite{2020:Rikos_Privacy_CDC} the authors present an event-based offset algorithm. 
%In this algorithm each node initially injects in its state a negative quantized offset. 
%Then, it injects a positive offset each time a set of event triggered conditions holds. 
This strategy allows the calculation of the exact quantized average in a finite number of time steps, but requires a large number of time steps for convergence. 
Finally, in \cite{2021:Rikos_Privacy_TCNS} the authors present an initial zero-sum offset algorithm. 
%In this algorithm each node injects a quantized offset to its state and to the states of its neighboring nodes only during the initialization steps. 
This strategy leads to fast finite time convergence to the exact average, but requires multiple simultaneous transmissions which increase significantly the header of the transmitted message. 

\vspace{.2cm}

\noindent
\textbf{Main Contributions.}
In this paper, we present a novel privacy preserving event-triggered distributed algorithm which (i) achieves average consensus under privacy constraints with quantized communication, (ii) converges after a finite number of time steps, and (iii) relies on event-driven operation and ceases transmissions once convergence has been achieved (which makes it suitable for battery powered or energy harvesting wireless networks). 
The main contributions of our paper are the following.
\begin{list4}
\item We present a novel privacy preserving distributed event-triggered algorithm which operates with quantized values and calculates the exact average of the initial states under privacy constraints; see Algorithm~\ref{algorithm_max} in Section~\ref{sec:Priv_Alg}. 
\item We show that our proposed privacy preserving algorithm converges after a finite number of iterations for which we provide a polynomial upper bound. Furthermore, we show our algorithm's transmission stopping capabilities; see Theorem~\ref{PROP1_max} in Section~\ref{sec:Conv_analysis}. 
%\item We calculate an upper bound on the number of transmissions each node performs during the operation of our algorithm and its privacy preserving mechanism; see Theorem~\ref{transm_bound_priv} in Section~\ref{sec:Trasm_analysis}. 
\item We present the necessary topological conditions that ensure privacy preservation for the nodes that follow the proposed algorithm; see Proposition~\ref{prop:1} in Section~\ref{sec:conditions}.
\item We demonstrate our algorithm's operation and compare its performance against other finite time privacy preserving algorithms from the current literature; see Section~\ref{results}. 
\end{list4}
The proposed privacy preserving algorithm relies on \textit{multiple state decomposition}. 
Specifically, at initialization, each node decomposes its initial state into multiple substate values. 
No substate value is equal to the initial state and each substate value is different than the other substate values. 
Furthermore, the average of the substate values is equal to the initial state. 
The node utilizes one substate as its initial state. 
Then, it injects the other substates to its state at specific instances and transmits them to a different node each time. 
This ensures that every neighboring node receives at least one substate.
Thus, the privacy of each node's initial state is preserved (at least if one neighboring node is not colluding with curious nodes).

The operation of our privacy preserving strategy relies on the quantized averaging algorithm in \cite{2021:Rikos_Finite_Trans_AUTO} (which is not privacy preserving). 
Unlike other privacy preserving strategies in the current literature (e.g., \cite{2013:Nikolas_Hadj, Mo-Murray:2017, 2019:Allerton_themis, 2021:Manitara_Rikos_Hadjicostis}), the proposed strategy takes full advantage of the finite time operation and transmission ceasing capability of \cite{2021:Rikos_Finite_Trans_AUTO}, without requiring global network parameters. 
This is mainly because the substate values are integers and the privacy preserving strategy converts the event trigger conditions in \cite{2021:Rikos_Finite_Trans_AUTO}. 
A comparison with current works is provided in Section~\ref{sec:compar_subsec}.

%\subsection{Paper Organization}

%The remainder of this paper is organized as follows. 
%In Section~\ref{preliminaries}, we introduce the notation used throughout the paper.
%In Section~\ref{probForm} we formulate the privacy preserving quantized average consensus problem.
%In Section~\ref{MaxAlgorithm}, we present our privacy preserving strategy and the privacy preserving event-triggered algorithm; we also analyze the algorithm's operation and establish its convergence. 
%We present an upper bound on the number of transmissions each node performs during the operation of the algorithm.
%Then, we present sufficient topological conditions that ensure privacy preservation. 
%In Section~\ref{results}, we present simulation results and compare the operation of the proposed algorithm against other algorithms in the current literature. 
%We conclude in Section~\ref{future} with a brief summary and remarks about future work.

% ===============================================
%
%
% NOTATION
%
%
% ===============================================
\section{NOTATION AND BACKGROUND}\label{preliminaries}

The sets of real, rational, integer and natural numbers are denoted by $ \mathbb{R}, \mathbb{Q}, \mathbb{Z}$ and $\mathbb{N}$, respectively. 
The set of nonnegative integers is denoted by $\mathbb{Z}_+$. 

%\subsection{Graph-Theoretic Notions}

\noindent
\textbf{Graph-Theoretic Notions.} 
Consider a network of $n$ ($n \geq 2$) agents communicating only with their immediate neighbors. 
The communication topology can be captured by a directed graph (digraph), called \textit{communication digraph}. 
A digraph is defined as $\mathcal{G}_d = (\mathcal{V}, \mathcal{E})$, where $\mathcal{V} =  \{v_1, v_2, \dots, v_n\}$ is the set of nodes and $\mathcal{E} \subseteq \mathcal{V} \times \mathcal{V} - \{ (v_j, v_j) \ | \ v_j \in \mathcal{V} \}$ is the set of edges (self-edges excluded). 
A directed edge from node $v_i$ to node $v_j$ is denoted by $m_{ji} \triangleq (v_j, v_i) \in \mathcal{E}$, and captures the fact that node $v_j$ can receive information from node $v_i$ (but not the other way around). 
We assume that the given digraph $\mathcal{G}_d = (\mathcal{V}, \mathcal{E})$ is \textit{strongly connected} (i.e., for each pair of nodes $v_j, v_i \in \mathcal{V}$, $v_j \neq v_i$, there exists a directed \textit{path}\footnote{A directed \textit{path} from $v_i$ to $v_j$ exists if we can find a sequence of vertices $v_i \equiv v_{l_0},v_{l_1}, \dots, v_{l_t} \equiv v_j$ such that $(v_{l_{\tau+1}},v_{l_{\tau}}) \in \mathcal{E}$ for $ \tau = 0, 1, \dots , t-1$.} from $v_i$ to $v_j$). 
The subset of nodes that can directly transmit information to node $v_j$ is called the set of in-neighbors of $v_j$ and is represented by $\mathcal{N}_j^- = \{ v_i \in \mathcal{V} \; | \; (v_j,v_i)\in \mathcal{E}\}$, while the subset of nodes that can directly receive information from node $v_j$ is called the set of out-neighbors of $v_j$ and is represented by $\mathcal{N}_j^+ = \{ v_l \in \mathcal{V} \; | \; (v_l,v_j)\in \mathcal{E}\}$. 
The cardinality of $\mathcal{N}_j^-$ is called the \textit{in-degree} of $v_j$ and is denoted by $\mathcal{D}_j^- = | \mathcal{N}_j^- |$, while the cardinality of $\mathcal{N}_j^+$ is called the \textit{out-degree} of $v_j$ and is denoted by $\mathcal{D}_j^+ = | \mathcal{N}_j^+ |$.

%\subsection{Node Operation}

\noindent
\textbf{Node Operation.} 
With respect to quantization of information flow, we have that at time step $k \in \mathbb{Z}_+$, each node $v_j \in \mathcal{V}$ maintains (i) the state variables $y^s_j[k], z^s_j[k], q_j^s[k]$ (where $y^s_j[k] \in \mathbb{Z}$, $z^s_j[k] \in \mathbb{Z}_+$, $q_j^s[k] = \frac{y_j^s[k]}{z_j^s[k]}$), (ii) the mass variables $y_j[k], z_j[k]$, (where $y_j[k] \in \mathbb{Z}$ and $z_j[k] \in \mathbb{Z}_+$), (iii) the substate counter $s_j$ (where $s_j \in \mathbb{N}$), (iv) the privacy variables $u_j^y[s_j]$, $u_j^z[s_j]$ (where $u_j^y[s_j] \in \mathbb{Z}$, $u_j^z[s_j] \in \mathbb{Z}$), (v) the transmission variables $S\_br_j$ and $M\_tr_j$ (where $S\_br_j \in \mathbb{N}$ and $M\_tr_j \in \mathbb{N}$). 
%The aggregate states are denoted by $y^s[k] = [y^s_1[k] \ ... \ y^s_n[k]]^{\rm T} \in \mathbb{Z}^n$, $z^s[k] = [z^s_1[k] \ ... \ z^s_n[k]]^{\rm T} \in \mathbb{Z}_+^n$, $q^s[k] = [q^s_1[k] \ ... \ q^s_n[k]]^{\rm T} \in \mathbb{Q}^n$ and $y[k] = [y_1[k] \ ... \ y_n[k]]^{\rm T} \in \mathbb{Z}^n$, $z[k] = [z_1[k] \ ... \ z_n[k]]^{\rm T} \in \mathbb{Z}_+^n$ respectively. 
Note here that for every node $v_j$, the state variables $y^s_j[k], z^s_j[k], q_j^s[k]$ are used to store the received messages and calculate the quantized average of the initial values, the mass variables $y_j[k], z_j[k]$ are used to communicate with other nodes by either transmitting or receiving messages, the substate counter $s_j$ is used to transmit the privacy variables, the privacy variables $u_j^y[s_j]$, $u_j^z[s_j]$ are used to preserve the privacy of the initial state, and the transmission variables $S\_br_j$, $M\_tr_j$ are used to decide whether the state variables will be broadcasted or the mass variables will be transmitted.

Furthermore, we assume that each node is aware of its out-neighbors and can directly (or indirectly\footnote{Indirect transmission could involve broadcasting a message to all out-neighbors while including in the message header the ID of the out-neighbor it is intended for.}) transmit messages to each out-neighbor; however, it cannot necessarily receive messages (at least not directly) from them. 
In the proposed distributed algorithm, each node $v_j$ assigns a \textit{unique order} in the set $\{0,1,..., \mathcal{D}_j^+ -1\}$ to each of its outgoing edges $m_{lj}$, where $v_l \in \mathcal{N}^+_j$. 
More specifically, the order of link $(v_l,v_j)$ for node $v_j$ is denoted by $P_{lj}$ (such that $\{P_{lj} \; | \; v_l \in \mathcal{N}^+_j\} = \{0,1,..., \mathcal{D}_j^+ -1\}$). 
This unique predetermined order is used during the execution of the proposed distributed algorithm as a way of allowing node $v_j$ to transmit messages to its out-neighbors in a \textit{round-robin}\footnote{When executing the proposed algorithm, each node $v_j$ transmits to its out-neighbors, one at a time, by following a predetermined order. The next time it transmits to an out-neighbor, it continues from the outgoing edge it stopped the previous time and cycles through the edges in a round-robin fashion according to the predetermined ordering.} fashion.

\section{PROBLEM FORMULATION}\label{probForm}

Consider a strongly connected digraph $\mathcal{G}_d = (\mathcal{V}, \mathcal{E})$, where each node $v_j \in \mathcal{V}$ has an initial (i.e., for $k=0$) quantized value $y_j[0]$ (for simplicity, we take $y_j[0] \in \mathbb{Z}$). 
Furthermore, consider that the node set $\mathcal{V}$ is partitioned into three subsets. 
Specifically, we have (i) the subset of nodes $v_j \in \mathcal{V}_p \subset \mathcal{V}$ that wish to preserve their privacy by not revealing their initial states $y_j[0]$ to other nodes, (ii) the subset of nodes $v_c \in \mathcal{V}_c \subset \mathcal{V}$ that are curious and try to identify the initial states $y[0]$ of all or a subset of nodes in the network, and (iii) the rest of the nodes  $v_i \in \mathcal{V}_n \subset \mathcal{V}$ that neither wish to preserve their privacy nor identify the states of any other nodes.
We assume that $\mathcal{V}_p \cap \mathcal{V}_c = \emptyset$, which means that curious nodes in $\mathcal{V}_c$ collaborate arbitrarily in order to identify the initial states of other nodes in the network. 
An example is shown in Fig.~\ref{prob_form_graph} (from \cite{2021:Rikos_Privacy_TCNS}).
\begin{figure}[h]
\begin{center}
\includegraphics[width=0.5\columnwidth]{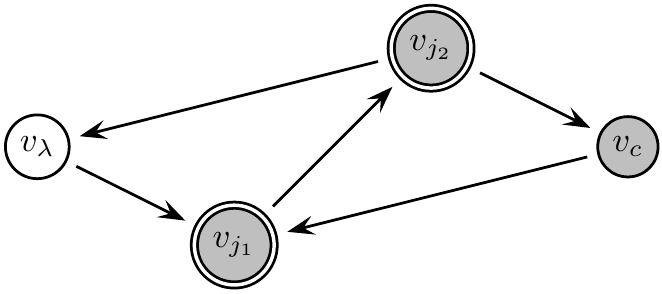}
\vspace{-0.2cm}
\caption{Example of a digraph with the different types of nodes in the network: nodes $v_{j_1}, v_{j_2} \in \mathcal{V}_p$ wish to preserve their privacy, node $v_{c} \in \mathcal{V}_c$ is curious and wishes to identify the initial states of other nodes in the network, and node $v_{\lambda} \in \mathcal{V}_{n}$ is neither curious nor wishes to preserve its privacy\vspace{-0.32cm}.}
\label{prob_form_graph}
\end{center}
\end{figure}

%\begin{remark}
%In this paper we assume that curious nodes $v_c$ collaborate arbitrarily in order to identify the initial states of other nodes in the network (i.e, $\mathcal{V}_p \cap \mathcal{V}_c = \emptyset$). 
%Note here that the case where $\mathcal{V}_p \cap \mathcal{V}_c \neq \emptyset$ is an easier case than $\mathcal{V}_p \cap \mathcal{V}_c = \emptyset$. 
%Specifically, the case $\mathcal{V}_p \cap \mathcal{V}_c \neq \emptyset$ means that nodes wish to preserve the privacy of their initial state against curious nodes that do \textit{not} collaborate, and do not share information between each other. 
%This case can be analyzed by considering multiple instances of our case (i.e., $\mathcal{V}_p \cap \mathcal{V}_c = \emptyset$). 
%\end{remark}

Privacy is defined as the ability of an individual node to seclude itself and thereby express itself selectively. 
We consider that the information of interest for each node is its initial state $y_j[0]$. 
We adopt the following notion of privacy, which aims to ensure that the state $y_j[0]$ cannot be inferred exactly by curious nodes and relates to notions of possible innocence in theoretical computer science \cite{1998:Reiter, 2006:Chatzikokolakis} in the sense that there is some uncertainty about $y_j[0]$.

\begin{defn}\label{Definition_Quant_Privacy} 
A node $v_j \in \mathcal{V}_p$ preserves the privacy of its initial state $y_j[0] \in \mathbb{Z}$ if $y_j[0]$ cannot be inferred by curious nodes $v_{c} \in \mathcal{V}_c$ at any point during the operation of the algorithm. 
This means that curious nodes in $\mathcal{V}_c$ cannot determine a finite range  $[\alpha, \beta]$ (where $\alpha < \beta$ and $\alpha, \beta \in \mathbb{R}$) in which the initial state $y_j[0]$ lies in.
\end{defn}

In this paper, we develop a distributed algorithm that allows nodes to address the problems \textbf{P1}, \textbf{P2} and \textbf{P3} presented below, while processing and transmitting \textit{quantized} information via available communication links. 

\noindent
\textbf{P1}. Every node $v_j$ obtains, after a finite number of steps, a fraction $q_j^s$ which is equal to the \textit{exact} average $q$ of the initial states of the nodes (i.e., there is no quantization error), where
\begin{equation}\label{initial_average}
q = \frac{\sum_{l=1}^{n}{y_l[0]}}{n} .
\end{equation}
Specifically, we argue that there exists $k_0$ so that for every $k \geq k_0$ we have 
\begin{equation}\label{alpha_z_y}
y^s_j[k] = \frac{\sum_{l=1}^{n}{y_l[0]}}{\alpha}  \ \ \text{and} \ \ z^s_j[k] = \frac{n}{\alpha} ,
\end{equation}
where $\alpha \in \mathbb{N}$. This means that 
\begin{equation}\label{alpha_q}
q^s_j[k] = \frac{(\sum_{l=1}^{n}{y_l[0]}) / \alpha}{n / \alpha} \coloneqq q ,
\end{equation}
for every $v_j \in \mathcal{V}$ (i.e., for $k \geq k_0$ every node $v_j$ has calculated $q$ as the ratio of two integer values). 

\noindent
\textbf{P2}. Every node $v_j \in \mathcal{V}_p$ preserves the privacy of its initial state $y_j[0]$ (i.e., it does not reveal its initial state $y_j[0]$ to other nodes) when it exchanges quantized information with neighboring nodes while calculating $q$ in \eqref{initial_average} (i.e., its state variables $y^s_j$, $z^s_j$, $q^s_j$ fulfill \eqref{alpha_z_y} and \eqref{alpha_q}, respectively). 

\noindent
\textbf{P3}. Every node $v_j$ stops performing transmissions towards its out-neighbors $v_l \in \mathcal{N}^+_j$ once its state variables $y^s_j$, $z^s_j$, $q^s_j$ fulfill \eqref{alpha_z_y} and \eqref{alpha_q}, respectively.

% ===============================================
%
%
% ALGORITHM FINITE TRANSMISSION EXACT QUANTIZED AVERAGING
%
%
% ===============================================
\section{PRIVACY PRESERVING EVENT-TRIGGERED QUANTIZED AVERAGE CONSENSUS ALGORITHM WITH FINITE TRANSMISSION CAPABILITIES}\label{MaxAlgorithm}

In this section we present a distributed algorithm which addresses problems \textbf{(P1)}, \textbf{(P2)}, \textbf{(P3)} presented in Section~\ref{probForm}. 
Before presenting the main functionalities of our algorithm, we make the following assumption.

\begin{assump}\label{assump_max_neigh_knowledge}
We assume that each node $v_j \in \mathcal{V}$ has knowledge of the maximum out-degree in the network $\mathcal{D}_{max}^+ = \max_{v_i \in \mathcal{V}} \mathcal{D}_j^+$. 
\end{assump}

Assumption~\ref{assump_max_neigh_knowledge} is important for guaranteeting convergence to the average of the initial states. 
In case Assumption~\ref{assump_max_neigh_knowledge} does not hold, then our algorithm may converge to a value that is not equal to the average of the initial states (i.e., our algorithm will simply achieve consensus).

\subsection{Initialization for Privacy Preserving Algorithm with Multiple State Decomposition}

Our strategy is based on the event-triggered deterministic algorithm in \cite{2021:Rikos_Finite_Trans_AUTO} with some modifications (since the algorithm in \cite{2021:Rikos_Finite_Trans_AUTO} is not privacy preserving). 
The main difference is the deployment of a mechanism that decomposes the initial state $y_j[0]$ of each node $v_j \in \mathcal{V}_p$ into $\mathcal{D}_{max}^+ + 2$ substates. 
The average of the $\mathcal{D}_{max}^+ + 2$ substates is equal to the initial state $y_j[0]$. 
Then, each substate is transmitted to a different out-neighbor at a different time step thus, effectively preserving the privacy of the initial state $y_j[0]$. 

In previous works (e.g., \cite{2013:Nikolas_Hadj, Mo-Murray:2017, 2019:Allerton_themis, 2020:Rikos_Privacy_CDC, 2021:Rikos_Privacy_TCNS}), each node $v_j \in \mathcal{V}_p$ injects a nonzero offset $u_j$ to its initial state. 
This means that it sets $\widetilde{y}_j[0] = y_j[0] + u_j$, where $u_j \neq 0$. 
However, in our case we require each node $v_j \in \mathcal{V}_p$ to decompose its initial state $y_j[0]$ into $\mathcal{D}_{max}^+ + 2$ substates whose average is equal to the initial state. 
Furthermore, each node $v_j$ maintains its substate counter $s_j \in \mathbb{N}$, and its privacy variables $u_j^y[s_j] \in \mathbb{Z}$, $u_j^z[s_j] \in \mathbb{Z}$. 
At initialization, each node $v_j \in \mathcal{V}_p$ chooses the privacy variables $u_j^y[s_j] \in \mathbb{Z}$, $u_j^z[s_j] \in \mathbb{Z}$, for $s_j \in \{ 0, 1, 2, ..., \mathcal{D}_{max}^+ + 1 \}$, to satisfy the
following constraints: 
\begin{subequations}\label{Offset_values}
\begin{align}
u_j^y[s_j] &\neq u_j^y[s_j'], \ u_j^y[s_j] \neq y_j[0], \nonumber \\ 
\forall s_j, s_j' & \in [0, \mathcal{D}_{max}^+ + 1], \ \text{with} \ s_j \neq s_j' \label{Offset_value_1bb} \\
u_j^y[s_j] &= 0, \ \forall \ s_j > \mathcal{D}_{max}^+ + 1, \label{Offset_value_1bbb} \\
u_j^z[s_j] &= 1, \ \forall \ s_j \in [0, \mathcal{D}_{max}^+ + 1], \label{Offset_value_1d} \\
u_j^z[s_j] &= 0, \ \forall \ s_j > \mathcal{D}_{max}^+ + 1, \label{Offset_value_1dd} \\
& \frac{\sum_{s_j = 0}^{\mathcal{D}_{max}^+ + 1} u_j^y[s_j]}{\mathcal{D}_{max}^+ + 2}= y_j[0].  \label{Offset_value_1e} 
\end{align}
\end{subequations}
Constraints \eqref{Offset_value_1bb}--\eqref{Offset_value_1e} are explicitly analyzed below: 
\\ \noindent
\textbf{1.} In \eqref{Offset_value_1bb} each substate $u_j^y[s_j]$ needs to have different value than every other substate $u_j^y[s_j']$ for $s_j \neq s_j'$ and $s_j, s_j' \in [0, \mathcal{D}_{max}^+ + 1]$. 
Furthermore, each substate $u_j^y[s_j]$ needs to have a different value than the initial state $y_j[0]$. 
This is important for not revealing the value of the initial state $y_j[0]$ and thus preserving its privacy. 
\\ \noindent
\textbf{2.} In \eqref{Offset_value_1bbb} each node $v_j$ stops injecting nonzero offsets after $\mathcal{D}_{max}^+ + 2$ time steps in order not to intervene with the calculation of the quantized average. 
This allows each node to calculate the exact quantized average of the initial states without any error.
\\ \noindent
\textbf{3.} In \eqref{Offset_value_1d} the substate $u_j^z[s_j]$ which is injected to the network by node $v_j$ needs to be equal to $1$ so that (i) the event-triggered conditions of the presented algorithm hold and (ii) the operation of the algorithm leads to the calculation of the exact average. 
\\ \noindent
\textbf{4.} In \eqref{Offset_value_1dd} each node $v_j$ stops injecting nonzero offsets after $\mathcal{D}_{max}^+ + 1$ time steps which allows the calculation of the quantized average without any error. 
\\ \noindent
\textbf{5.} In \eqref{Offset_value_1e} the average of the total injected offset in the network by node $v_j$ needs to be equal to node $v_j$'s initial state $y_j[0]$. 
This means that each node $v_j$ creates $\mathcal{D}_{max}^+ + 2$ substates of its initial state which have different $u_j^y[s_j]$ values.  
These substates allow calculation of the exact quantized average of the initial states without any error. 

The above choices imply that each node $v_j \in \mathcal{V}_p$ generates $\mathcal{D}_{max}^+ + 2$ substates $u_j^y[s_j]$, $u_j^z[s_j]$ of its initial state $y_j[0]$, for which it holds (i) $u_j^y[s_j] \neq u_j^y[s_j']$, $s_j \neq s_j'$ and (ii) $u_j^y[s_j] \neq y_j[0]$, for every $s_j, s_j' \in [0, \mathcal{D}_{max}^+ + 1]$. 
Then, during every time step $s_j \in [0, \mathcal{D}_{max}^+ + 1]$, node $v_j$ injects in the network the substates $u_j^y[s_j]$, $u_j^z[s_j]$. 
This leads to the calculation of the exact quantized average in a privacy preserving manner. 
Note that if $u_j^y[s_j] \neq y_j[0]$, for every $s_j \in [0, \mathcal{D}_{max}^+ + 1]$, does not hold, then the exact quantized average may be calculated but not in a privacy preserving manner. 

\begin{remark}
Note here that each node $v_j \in \mathcal{V}_p$ that wishes to preserve its privacy chooses its privacy variables $u_j^y[s_j] \in \mathbb{Z}$, $u_j^z[s_j] \in \mathbb{Z}$ according to \eqref{Offset_value_1bb}--\eqref{Offset_value_1e}. 
However, each node $v_i \notin \mathcal{V}_p$ which does not wish to preserve its privacy sets $u_i^y[s_j] = y_i[0]$, for every $s_i \in [0, \mathcal{D}_{max}^+ + 1]$ and also chooses its privacy variables $u_i^y[s_j] \in \mathbb{Z}$, $u_i^z[s_j] \in \mathbb{Z}$, according to \eqref{Offset_value_1bbb}--\eqref{Offset_value_1e} (i.e., \eqref{Offset_value_1bb} does not hold).

\end{remark}

\subsection{Privacy Preserving Finite Transmission Event-Triggered Algorithm}\label{sec:Priv_Alg}

%The proposed privacy preserving algorithm is a quantized value process and transmission mechanism. 
%In the algorithm each node $v_j \in \mathcal{V}$ in a strongly connected digraph $\mathcal{G}_d = (\mathcal{V}, \mathcal{E})$ processes and transmits quantized values according to multiple event-triggered conditions. 
%The main idea is that every node incorporates its substates to its mass variables while it maintains separate event-triggering mechanisms for (i) broadcasting its state variables, and (ii) transmitting its mass variables. 
%This way, nodes learn the average in a privacy preserving manner, but also have a way to decide when (or not) to transmit. 
The details of the privacy preserving distributed algorithm with transmission stopping can be seen in Algorithm~\ref{algorithm_max}. 

\noindent
\begin{varalgorithm}{1}
\caption{Privacy Preserving Finite Transmission Event-Triggered Quantized Average Consensus}
\textbf{Input:} A strongly connected digraph $\mathcal{G}_d = (\mathcal{V}, \mathcal{E})$ with $n=|\mathcal{V}|$ nodes and $m=|\mathcal{E}|$ edges. Each node $v_j\in \mathcal{V}$ has an initial state $y_j[0] \in \mathbb{Z}$ and has knowledge of $\mathcal{D}_{max}^+$. 
\\
\textbf{Initialization:} Each node $v_j \in \mathcal{V}$ does the following: 
\begin{list4}
\item[1)] Assigns to each outgoing edge $v_l \in \mathcal{N}^+_j$ a unique order $P_{lj}$ in the set $\{0,1,..., \mathcal{D}_j^+ -1\}$. 
\item[2)] Sets counter $s_j = 0$. Chooses $u^y_j[k] \neq u^y_j[k_1]$, $u^y_j[k] \neq y_j[0]$ where $k \neq k_1$ for every $k, k_1 \in \{0, 1, \ldots, \mathcal{D}_{max}^+ + 1\}$, and $u^y_j[k'] = 0$ for $k' > \mathcal{D}_{max}^+ + 1$ such that $(\sum_{s_j = 0}^{\mathcal{D}_{max}^+ + 1} u_j^y[s_j]) / (\mathcal{D}_{max}^+ + 2) = y_j[0]$. 
\item[3)] Chooses $u^z_j[k] = 1$ for $k \in \{0, 1, \ldots, \mathcal{D}_{max}^+ + 1\}$, and $u^z_j[k'] = 0$ for $k' > \mathcal{D}_{max}^+ + 1$. 
\item[4)] Sets $y_j[0] = u^y_j[s_j]$, $z_j[0] = u^z_j[s_j]$, $z^s_j[0] = z_j[0]$, $y^s_j[0] = y_j[0]$, $q^s_j[0] = y^s_j[0] / z^s_j[0]$, $s_j = s_j + 1$ and $S\_br_j = 0$, $M\_tr_j = 0$. 
\item[5)] Broadcasts $z^s_j[0]$, $y^s_j[0]$ to every $v_l \in \mathcal{N}_j^+$.
\end{list4}
\textbf{Iteration:} For $k=0,1,2,\dots$, each node $v_j \in \mathcal{V}$: 
\begin{list4}
\item[1)] Receives $y^s_i[k]$, $z^s_i[k]$ from every $v_i \in \mathcal{N}_j^-$ (if no message is received it sets $y^s_i[k] = 0$, $z^s_i[k] = 0$). 
\item[2)] Receives $y_i[k]$, $z_i[k]$ from each $v_i \in \mathcal{N}_j^-$ and sets 
$$
y_j[k+1] = y_j[k] + \sum_{v_i \in \mathcal{N}_j^-} w_{ji}[k]y_i[k] ,
$$ 
$$
z_j[k+1] = z_j[k] + \sum_{v_i \in \mathcal{N}_j^-} w_{ji}[k]z_i[k] ,
$$
where $w_{ji}[k]=1$ if a message with $y_i[k]$, $z_i[k]$ is received from in-neighbor $v_i$, otherwise $w_{ji}[k]=0$. 
\item[3)] \textbf{If} $w_{ji}[k] \neq 0$ or $z^s_i[k] \neq 0$ for some $v_i \in \mathcal{N}_j^-$ \textbf{then} calls Algorithm~\ref{algorithm_max_1a}
\item[4)] Sets $M\_tr_j = \max\{M\_tr_j, u^z_j[s_j]\}$. 
\item[5)] \textbf{If} $M\_tr_j = 1$ \textbf{then} (i) sets $y_j[k] = y_j[k] + u^y_j[s_j]$, $z_j[k] = z_j[k] + u^z_j[s_j]$ and (ii) chooses $v_l \in \mathcal{N}_j^+$ according to $P_{lj}$ (in a round-robin fashion) and transmits $y_j[k]$, $z_j[k]$. 
Then, sets $y_j[k] = 0$, $z_j[k] = 0$, $M\_tr_j = 0$, $s_j = s_j + 1$. 
\item[6)] \textbf{If} $S\_br_j = 1$ \textbf{then} broadcasts $z^s_j[k+1]$, $y^s_j[k+1]$ to every $v_l \in \mathcal{N}_j^+$. 
Then, sets $S\_br_j = 0$. 
\item[7)] Repeats (increases $k$ to $k + 1$ and goes to Step~$1$).
\end{list4}
\textbf{Output:} \eqref{alpha_q} holds for every $v_j \in \mathcal{V}$. 
\label{algorithm_max}
\end{varalgorithm}

\noindent
\vspace{-0.5cm}    
\begin{varalgorithm}{1.A}
\caption{Event-Triggered Conditions for Algorithm~\ref{algorithm_max} (for each node $v_j$)}
\textbf{Input} 
\\ $y^s_j[k]$, $z^s_j[k]$, $q^s_j[k]$, $y_j[k+1]$, $z_j[k+1]$, $S\_br_j$, $M\_tr_j$ and the received $y^s_i[k]$, $z^s_i[k]$ from every $v_i \in \mathcal{N}_j^-$.
\\
\textbf{Execution} 
\begin{list4}
\item[1)] \underline{Event Trigger Conditions~$1$:} \textbf{If} 
\\ Condition~$(i)$: $z^s_i[k] > z^s_j[k]$, or
\\ Condition~$(ii)$: $z^s_i[k] = z^s_j[k]$ and $y^s_i[k] > y^s_j[k]$, 
\\ \textbf{then} sets 
$$ 
z^s_j[k+1] = \max_{v_i \in \mathcal{N}_j^-} z^s_i[k] , \ \ \text{and}
$$ 
$$ 
y^s_j[k+1] = \max_{v_i \in \{v_{i'} \in \mathcal{N}_j^- | z^s_{i'}[k] = z^s_j[k+1]\}} y^s_i[k] ,
$$ 
and sets $q^s_j[k+1] = \frac{y^s_j[k+1]}{z^s_j[k+1]}$, and $S\_br_j = 1$. 
\item[2)] \underline{Event Trigger Conditions~$2$:} \textbf{If}
\\ Condition~$(i)$: $z_j[k+1] > z^s_j[k+1]$, or 
\\ Condition~$(ii)$: $z_j[k+1] = z^s_j[k+1]$ and $y_j[k+1] > y^s_j[k+1]$, 
\\ \textbf{then} sets $z^s_j[k+1] = z_j[k+1]$, $y^s_j[k+1] = y_j[k+1]$ and sets $q^s_j[k+1] = \frac{y^s_j[k+1]}{z^s_j[k+1]}$ and $S\_br_j = 1$.
\item[3)] \underline{Event Trigger Conditions~$3$:} \textbf{If}
\\ Condition~$(i)$: $0 < z_j[k+1] < z^s_j[k+1]$ or 
\\ Condition~$(ii)$: $z_j[k+1] = z^s_j[k+1]$ and $y_j[k+1] < y^s_j[k+1]$, 
\\ \textbf{then} sets $M\_tr_j = 1$.
\end{list4}
\textbf{Output} 
\\ $y^s_j[k]$, $z^s_j[k]$, $q^s_j[k]$, $S\_br_j$, $M\_tr_j$.
\label{algorithm_max_1a}
\end{varalgorithm}

The intuition behind Algorithm~\ref{algorithm_max} is the following. 
Let us first consider the notion of ``leading mass''.
During time step $k$, the set of mass variables which has the largest $z[k]$ value is the ``leading mass''. 
In case there are multiple sets of mass variables that have the largest $z[k]$, then the ``leading mass'' is the set of mass variables that has the largest $y[k]$ value among the sets of mass variables with the largest $z[k]$. 
Note that a formal definition of the ``leading mass'' is presented in Section~\ref{sec:Conv_analysis}. 
Each node $v_j \in \mathcal{V}_p$ that would like to preserve its privacy performs the following steps: \\
\noindent
\textbf{Initialization.} \\ \noindent
\textbf{A.} Node $v_j$ assigns to each outgoing edge a unique order $P_{lj}$ in order to perform transmissions in a round-robin fashion.  
\\ \noindent
\textbf{B.} Node $v_j$ initializes the substate counter $s_j$ to zero (i.e., $s_j=0$) and the set of $(\mathcal{D}_{max}^+ + 2)$ privacy variables $u_j^y[s_j]$, $u_j^z[s_j]$ according to \eqref{Offset_value_1bb}--\eqref{Offset_value_1e} for $s_j \in [0, \mathcal{D}_{max}^+ + 1]$. 
For example, suppose that node $v_j$ has initial state $y_j[0] = 4$ and the maximum out-degree in the network is equal to $3$. This means that it decomposes its initial state in $5$ substates. Specifically, it can (randomly) set $u_j^y[0] = 1$, $u_j^y[1] = 8$, $u_j^y[2] = 6$, $u_j^y[3] = 2$, $u_j^y[4] = 3$, and $u_j^z[0] = 1$, $u_j^z[1] = 1$, $u_j^z[2] = 1$, $u_j^z[3] = 1$, $u_j^z[4] = 1$. Note that the average of the substates $u_j^y[s_j]$, $s_j \in \{0,1,2,3,4\}$ is equal to $4$. 
\\ \noindent
\textbf{C.} Node $v_j$ utilizes the substates $u_j^y[0]$, $u_j^z[0]$ as its initial state (i.e., it sets $y_j[0] = u_j^y[0]$ and $z_j[0] = u_j^z[0]$). Then, considers its set of stored mass variables $y_j[0]$, $z_j[0]$ to be the ``leading mass''. For this reason, it sets its state variables $z^s_j[0]$, $y^s_j[0]$, $q^s_j[0]$ to be equal to the stored mass variables $y_j[0]$, $z_j[0]$, and then broadcasts the values of its state variables. 

\noindent
\textbf{Iteration.}
\\ \noindent
\textbf{A.} Node $v_j$ receives the (possibly) transmitted state variables from its in-neighbors and, (ii) receives and stores the (possibly) transmitted mass variables from its in-neighbors. 
\\ \noindent
\textbf{B.} If node $v_j$ received a set of state variables and/or a set of mass variables from its in-neighbors, then it executes Algorithm~\ref{algorithm_max_1a}. During Algorithm~\ref{algorithm_max_1a} each node checks: 
\\ 
\textbf{B -- Event Trigger Conditions~$1$:} It checks whether the received set of state variables is equal to the ``leading mass''. If it receives messages from multiple in-neighbors it checks which set of state variables is the ``leading mass''. If Event Trigger Conditions~$1$ hold, it sets its state variables to be equal to the received set of state variables which is the ``leading mass'' and decides to broadcast its updated state variables (i.e., sets its transmission variable $S\_br_j = 1$). 
\\
\textbf{B -- Event Trigger Conditions~$2$:} It checks whether the set of mass variables it stored is the ``leading mass.'' If this condition holds, it sets its state variables to be equal to the stored set of mass variables and decides to broadcast its updated state variables (i.e., sets its transmission variable $S\_br_j = 1$). 
\\ 
\textbf{B -- Event Trigger Conditions~$3$:} It checks whether the set of mass variables it stored is not the ``leading mass'' (i.e., it checks whether its state variables are equal to the ``leading mass''). If this condition holds, this means that the mass variables of another node in the network is the ``leading mass'' (and the state variables of node $v_j$ became equal to the ``leading mass'' from Event Trigger Conditions~$1$). 
This means the stored mass variables is not the ``leading mass'' and thus $v_j$ decides to transmit its stored mass variables (i.e., sets its transmission variable $M\_tr_j = 1$). 
\\ \noindent
\textbf{C.} Node $v_j$ sets its transmission variable $M\_tr_j$ to be equal to the maximum value of $M\_tr_j$ and the substate $u_j^z[s_j]$. This step is important for the privacy preserving mechanism. Note that the value of the substate $u_j^z[s_j]$ is equal to $1$ for $s_j \in [0, \mathcal{D}_{max}^+ + 1]$. This means that the value of the transmission variable $M\_tr_j$ will become equal to $1$ for the first $\mathcal{D}_{max}^+ + 1$ time steps (since $s_j$ becomes equal to $1$ during the Initialization procedure). Thus, node $v_j$ will perform transmissions of the substates $u_j^y[s_j]$, $u_j^z[s_j]$ towards its out-neighbors. 
\\ \noindent
\textbf{D.} If $M\_tr_j$ is equal to $1$, node $v_j$ (i) injects the substates $u_j^y[s_j]$, $u_j^z[s_j]$ to its mass variables, (ii) transmits its mass variables towards an out-neighbor according to the unique order $P_{lj}$, and (iii) increases the substate counter $s_j$. 
\\ \noindent
\textbf{E.} If $S\_br_j$ is equal to $1$ node $v_j$ broadcasts its state variables towards every out-neighbor. Then, it repeats the procedure. 

% \vspace{.3cm}

% Note again that the above steps are executed from each node that wishes to preserve its privacy (i.e., $v_j \in \mathcal{V}_p$ executes Algorithm~\ref{algorithm_max}). 
% However, each node $v_i \notin \mathcal{V}_p$ that does not wish to preserve its privacy executes Algorithm~\ref{algorithm_max} but sets $u_i^y[s_j] = y_j[0]$, for every $s_j \in [0, \mathcal{D}_{max}^+ + 1]$, during Initialization Step~$2$. 

\subsection{Comparison with Previous Works}\label{sec:compar_subsec}

%\begin{remark}\label{compar_remark}
Algorithm~\ref{algorithm_max} is significantly different than most asymptotic or finite time algorithms in the current literature (see, e.g., \cite{2013:Nikolas_Hadj, Mo-Murray:2017, 2019:Allerton_themis, 2020:Rikos_Privacy_CDC, 2021:Rikos_Privacy_TCNS}). Specifically, 
Algorithm~\ref{algorithm_max} is different from asymptotic algorithms (see, e.g., \cite{2013:Nikolas_Hadj,Mo-Murray:2017, 2019:Allerton_themis}). 
These algorithms are able to calculate the average of the initial states in an asymptotic fashion which includes transmission of real valued messages and results in asymptotic convergence. 
Furthermore, these algorithms do not provide finite transmission guarantees. 
Algorithm~\ref{algorithm_max} is also different from the finite time algorithms in \cite{2020:Rikos_Privacy_CDC, 2021:Rikos_Privacy_TCNS}. 
In the event-based offset algorithm in \cite{2020:Rikos_Privacy_CDC}, each node $v_j$ initially injects a negative valued quantized offset to its state. 
Then, it injects a positive valued offset to its state only when its event-triggered conditions hold. 
This strategy relies on transmission of quantized values towards a single out-neighbor (which does not increase the header of the transmitted message). 
However, it requires a significant number of time steps for convergence (since the event-triggered conditions of each node need to hold for a specific number of instances before the total offset is injected in the network).  Also, \cite{2020:Rikos_Privacy_CDC} does not exhibit finite transmission guarantees and is not suitable for finite transmission operation. 
Specifically, if a distributed stopping protocol is utilized, \cite{2020:Rikos_Privacy_CDC} may converge and stop before each node injects the total stored offset in the network (since each node injects a positive valued offset to its state only when its event-triggered conditions hold). 
This will result in achieving consensus but not equal to the average of the initial states. 
In the initial zero-sum offset algorithm in \cite{2021:Rikos_Privacy_TCNS}, each node injects a quantized offset to its out-neighboring nodes during the algorithm's initialization steps which leads to fast convergence. 
However, it requires multiple simultaneous transmissions of different quantized values towards every out-neighbor. 
This increases significantly the header of the transmitted message and each node keeps performing transmissions once convergence has been achieved. 
Furthermore, \cite{2021:Rikos_Privacy_TCNS} does not exhibit finite transmission guarantees. 
The privacy preserving strategy of Algorithm~\ref{algorithm_max} takes full advantage of the finite time nature of the underlying quantized averaging algorithm (which is presented in \cite{2021:Rikos_Finite_Trans_AUTO}). 
Specifically, each node decomposes its state into multiple substates and utilizes one substate as its initial state. 
Then, it transmits each one of the other substates to a single out-neighbor at specific time steps. 
This means that (i) consensus to the exact average of the initial states is reached after a finite number of iterations, (ii) no initial error is introduced which allows faster convergence of the algorithm (i.e., no offset is initially infused in the network by the nodes following the algorithm), (iii) each node performs directed transmissions towards a single out-neighbor or broadcast transmissions towards every out-neighbor (which does not increase significantly the header of the transmitted message), and (iv) every node ceases transmissions once the average of the initial states is calculated in a privacy preserving manner. 
%\end{remark}

\subsection{Convergence Analysis of Algorithm~\ref{algorithm_max}}\label{sec:Conv_analysis}

Before analyzing the deterministic convergence of Algorithm~\ref{algorithm_max}, we consider the following setup. 

{\it Setup:} Consider a strongly connected digraph $\mathcal{G}_d = (\mathcal{V}, \mathcal{E})$ with $n=|\mathcal{V}|$ nodes and $m=|\mathcal{E}|$ edges. 
During the execution of Algorithm~\ref{algorithm_max}, at time step $k_0$, there is at least one node $v_{j'} \in \mathcal{V}$, for which 
\begin{equation}\label{great_z_prop1_det}
z_{j'}[k_0] \geq z_i[k_0], \ \forall v_i \in \mathcal{V}.
\end{equation}
Then, among the nodes $v_{j'}$ for which (\ref{great_z_prop1_det}) holds, there is at least one node $v_j$ for which 
\begin{equation}\label{great_z_prop2_det}
y_j[k_0] \geq y_{l}[k_0] , \ \text{where} \ \ v_j, v_{l} \in \{ v_{j'} \in \mathcal{V} \ | \ (\ref{great_z_prop1_det}) \ \text{holds} \}.
\end{equation}
For notational convenience we will call the mass variables of node $v_j$ for which (\ref{great_z_prop1_det}) and (\ref{great_z_prop2_det}) hold as the ``leading mass'' (or ``leading masses'').

We now consider the following two lemmas which are necessary for our subsequent development. 
Due to space considerations we omit the proofs of the two lemmas; they will be available in an extended version of our paper. 

\begin{lemma}\label{before_second_lemma}
If, during time step $k_0$ of Algorithm~\ref{algorithm_max}, the mass variables of node $v_j$ fulfill (\ref{great_z_prop1_det}) and (\ref{great_z_prop2_det}), then the state variables of every node $v_i \in \mathcal{V}$ satisfy 
\begin{eqnarray}\label{first_z}
z_i^s[k_0] \leq z_j[k_0] ,
\end{eqnarray}
or 
\begin{eqnarray}\label{first_zy}
z_i^s[k_0] = z_j[k_0] \ \ \text{and} \ \ y_i^s[k_0] \leq y_j[k_0].
\end{eqnarray}
\end{lemma}

%\begin{proof}
%Let us consider the variable
%$$
%z^{(m)}[k] = \max_{v_l \in \mathcal{V}} z_l[k] . 
%$$
%From Iteration Step~$4$ of Algorithm~\ref{algorithm_max} we have that $z^{(m)}[k]$ is non-decreasing (i.e., $z^{(m)}[k+1] \geq z^{(m)}[k]$, for every $k$). 
%Furthermore, since the mass variables of node $v_j$ fulfill (\ref{great_z_prop1_det}) and (\ref{great_z_prop2_det}), then, during every time step $k$, it holds that  
%$$
%z_j[k] = z^{(m)}[k]. 
%$$

%In addition, for every $k$, during Iteration Steps~$2$ and $5$ of Algorithm~\ref{algorithm_max}, for every node $v_i \in \mathcal{V}$, we have that $z_i^s[k]$ is either less than $z^{(m)}[k]$ (i.e., $z_i^s[k] < z^{(m)}[k]$) or equal to $z^{(m)}[k]$ (i.e., $z_i^s[k] = z^{(m)}[k]$). 
%As a result, at time step $k$, the state variables of every node $v_i \in \mathcal{V}$ satisfy 
%$$
%z_i^s[k] \leq z_j[k] .
%$$

%Finally, from Iteration Steps~$2$ and $5$, for every $k$, we have that, if it holds that $z_i^s[k] = z_j[k]$ for some node $v_i$ then we have that either $y_i^s[k] < y_j[k]$ or $y_i^s[k] = y_j[k]$. 
%[Note here that if $z_i^s[k] = z_j[k]$ and $y_i^s[k] > y_j[k]$, then the mass variables of $v_j$ do not fulfill (\ref{great_z_prop1_det}) and (\ref{great_z_prop2_det}) which is a contradiction.]
%As a result we have that if the mass variables of node $v_j$ fulfill (\ref{great_z_prop1_det}) and (\ref{great_z_prop2_det}), then the state variables of every node $v_i \in \mathcal{V}$ satisfy (\ref{first_z}) or (\ref{first_zy}). 
%\end{proof}

\begin{lemma}\label{second_lemma}
If, during time step $k_0$ of Algorithm~\ref{algorithm_max}, the mass variables of {\em each} node $v_j$ with nonzero mass variables fulfill (\ref{great_z_prop1_det}) and (\ref{great_z_prop2_det}), then we have only ``leading masses'' and no ``follower masses''. 
This means that the ``Event Trigger Conditions~$2$'' will never hold again for future time steps $k \geq k_0$. 
As a result, the transmissions that (may) take place will only be via broadcasting (from ``Event Trigger Conditions~$1$ and $3$'') for at most $n-1$ time steps and then they will cease. 
\end{lemma}

We now analyze the deterministic convergence of Algorithm~\ref{algorithm_max}. 
Due to space considerations we provide a sketch of the proof. 

\begin{theorem}\label{PROP1_max}
Consider a strongly connected digraph $\mathcal{G}_d = (\mathcal{V}, \mathcal{E})$ with $n=|\mathcal{V}|$ nodes and $m=|\mathcal{E}|$ edges. 
The execution of Algorithm~\ref{algorithm_max} allows each node $v_j \in \mathcal{V}$ to reach quantized average consensus after a finite number of time steps $k_0$ upper bounded by $1 + \mathcal{D}_{max}^+ + n^2 + (n-1)m^2$, where $n$ is the number of nodes and $m$ is the number of edges in the network, and $\mathcal{D}_{max}^+$ is the maximum out-degree in the network. 
Furthermore, each node stops transmitting towards its out-neighbors once quantized average consensus is reached. 
\end{theorem}

\begin{proof} 
In this proof we will show that there exists $k_0 \in \mathbb{Z}_+$ for which the mass variables of {\em every} node $v_j$ (for which $z_j[k] > 0$) fulfill (\ref{great_z_prop1_det}) and (\ref{great_z_prop2_det}), for every $k \geq k_0$. 
This means that for $k \geq k_0$ we have only leading masses. 
Note here that, during the execution of Algorithm~\ref{algorithm_max}, the leading mass will not fulfill the Event Trigger Conditions~$3$ (in Execution Step~$3$ of Algorithm~\ref{algorithm_max_1a}).
This means that the corresponding node (say $v_j$) will not transmit its mass variables to its out-neighbors $v_l \in \mathcal{N}_j^+$ according to its predetermined priority. 
Furthermore, from Lemma~\ref{second_lemma} we will also show that there exists $k_1 > k_0$ ($k_1 \in \mathbb{Z}_+$), where for every $k \geq k_1$ the state variables of every $v_j \in \mathcal{V}$ fulfill (\ref{alpha_z_y}) and (\ref{alpha_q}) for $\alpha \in \mathbb{Z}_+$ (i.e., every node has reached quantized consensus) and thus transmissions cease. 

During the Initialization steps of Algorithm~\ref{algorithm_max}, each node $v_j$ will (i) decompose its initial state $y_j[0]$ into $\mathcal{D}_{max}^+ + 2$ substates $u_j^y[s_j] \in \mathbb{Z}$, for $s_j \in \{ 0, 1, 2, ..., \mathcal{D}_{max}^+ + 1\}$, (ii) set its initial state $y_j[0]$ to be equal to $u_j^y[0]$, (iii) increase the substate counter $s_j$, and (iv) broadcast its state variables to every out-neighbor. 
Then, during Iteration Step~$1$, each node will (i) receive and update its state variables, (ii) receive and update its mass variables, and (iii) call Algorithm~\ref{algorithm_max_1a} to check Event Trigger Conditions~$1$, Event Trigger Conditions~$2$, and Event Trigger Conditions~$3$. 
However, note here that regardless of the output of Algorithm~\ref{algorithm_max_1a}, each node $v_j$ will utilize the privacy preseving strategy for the first $\mathcal{D}_{max}^+ + 1$ time steps (i.e., for $k = 0, 1, ..., \mathcal{D}_{max}^+$).
This means that for the first $\mathcal{D}_{max}^+ + 1$ time steps of the Iteration procedure, each node $v_j$ will inject to its mass variables the set of substates $u_j^y[s_j]$, $u_j^z[s_j]$, for $s_j \in \{ 0, 1, 2, ..., \mathcal{D}_{max}^+ \}$. 
Then, it will transmit its mass variables to an out-neighbor according to the order $P_{lj}$. 
Thus, after $\mathcal{D}_{max}^+ + 1$ time steps, each node $v_j$ will have injected in the network every set of substates $u_j^y[s_j]$, $u_j^z[s_j]$. 

For the analysis of the execution of Algorithm~\ref{algorithm_max} for time steps $k \geq \mathcal{D}_{max}^+ + 1$, we can use steps similar to the analysis in the proof of Theorem~$1$ in \cite{2021:Rikos_Finite_Trans_AUTO}. 
\end{proof}

\subsection{Topological Conditions for Privacy Preservation}\label{sec:conditions}

We now establish necessary topological conditions that ensure privacy for every $v_j \in \mathcal{V}_p$ which follows Algorithm~\ref{algorithm_max}. 

\begin{prop}\label{prop:1}
Consider a fixed strongly connected digraph $\mathcal{G}_d = (\mathcal{V}, \mathcal{E})$ with $n=|\mathcal{V}|$ nodes. 
Assume that a subset of nodes $v_j \in \mathcal{V}_p$ follow Algorithm~\ref{algorithm_max} where they choose the set of subsets chosen as in \eqref{Offset_value_1bb}--\eqref{Offset_value_1e}. 
Curious nodes $v_c \in \mathcal{V}_c$ will not be able to identify the initial state $y_j [0]$ of $v_j \in \mathcal{V}_p$, as long as $v_j$ has at least one in- or out-neighbor $v_{\ell}\in \mathcal{V}_p$ connected to it that aims to preserve its privacy. 
\end{prop}

\begin{proof}
In this proof we consider the following cases regarding the topological conditions for privacy preservation during the execution of Algorithm~\ref{algorithm_max}. 
Then, we summarize the results and derive the necessary and sufficient topological conditions for privacy preservation.  
\\ \noindent
\textbf{A.} Every in- and out-neighbor of node $v_j$ is curious (i.e., $v_i \in \mathcal{V}_c$, $\forall v_i \in \mathcal{N}_j^-$, and $v_l \in \mathcal{V}_c$, $\forall v_l \in \mathcal{N}_j^+$). 
In this case, the curious in- and out-neighbors communicate with each other and node $v_j \in \mathcal{V}_p$ will not be able to keep its privacy. 
Specifically, at Initialization curious nodes will know $u_j^y[0]$. 
Then, during the Iteration procedure, curious nodes will know the messages $v_j$ has received and the messages $v_j$ has transmitted. 
This means that they will be able to determine the values of $u_j^y[s_j] \in \mathbb{Z}$, for $s_j \in \{1, 2, ..., \mathcal{D}_{max}^+ \}$. 
Note that the average of every $u_j^y[s_j]$, for $s_j \in \{0, 1, 2, ..., \mathcal{D}_{max}^+ \}$, is equal to $v_j$'s initial state $y_j[0]$. 
This means that curious nodes will be able to determine the initial state $y_j[0]$. 
As a result, for the case where every in- and out-neighbor of node $v_j$ is curious, node $v_j$ does not preserve the privacy of its initial state. 
\\ \noindent
\textbf{B.} One out-neighbor of node $v_j$, say $v_{l'}$, is neither curious nor following the privacy preserving strategy (i.e., $v_{l'} \in \mathcal{V}_n$), and all other in- and out-neighbors of both nodes $v_j$, $v_{l'}$ are curious (i.e., $v_i \in \mathcal{V}_c$, $\forall v_i \in \mathcal{N}_j^-$, and $v_{l} \in \mathcal{V}_c$, $\forall v_l \in \mathcal{N}_j^+ \setminus \{ v_{l'} \}$). 
During the Initialization procedure, curious nodes will know $y_{l'}[0]$. 
Also, during the Iteration procedure, curious nodes will know the messages $v_j$ has received and the messages $v_j$ has transmitted. 
Furthermore, curious nodes can infer the input of node $v_l$ from its output.
Then, they will be able to extract the messages of node $v_j$ (as if a curious node was directly connected to node $v_j$). 
As a result, for the case where one out-neighbor of node $v_j$ is neither curious nor following the privacy preserving protocol and every other in- and out-neighbor of node $v_j$ is curious, $v_j$ does not preserve the privacy of its initial state. 
\\ \noindent
\textbf{C.} One out-neighbor of node $v_j$, say $v_{l'}$, is following the privacy preserving strategy (i.e., $v_{l'} \in \mathcal{V}_p$) and all other in- and out-neighbors of both nodes $v_j$, $v_{l'}$ are curious (i.e., $v_i \in \mathcal{V}_c$, $\forall v_i \in \mathcal{N}_j^-$, and $v_{l} \in \mathcal{V}_c$, $\forall v_l \in \mathcal{N}_j^+ \setminus \{ v_{l'} \}$). 
During the Iteration procedure, curious nodes will not be able to infer the substates transmitted from node $v_j$ to node $v_{l'}$. 
As a result, curious nodes will not be able to infer the initial state of node $v_j$ and the initial state of node $v_{l'}$. 
Thus, in this case node $v_j$ preserves the privacy of its initial state. 
\\ \noindent
\textbf{D.} The case where one in-neighbor of $v_j$, say $v_{i'}$, is following the privacy preserving strategy and every in- and out-neighbors of both nodes are curious can be analyzed as case \textbf{C}. 

From the four cases \textbf{A -- D} we considered, we have that a node $v_j \in \mathcal{V}_p$ is able to preserve its privacy if it has at least one in- or out-neighbor (say $v_{i'}$ or $v_{l'}$) who also wants to preserve its privacy and follows the proposed privacy preserving strategy. 
Furthermore, it is important to note that curious nodes will not be able to determine (i) the values of the messages transmitted from $v_{i'}$ to $v_j$, and (ii) the values of the messages transmitted from $v_j$ to $v_{l'}$. 
This means that curious nodes will not be able to determine a finite range  $[\alpha, \beta]$ (where $\alpha < \beta$ and $\alpha, \beta \in \mathbb{R}$) in which the initial state $y_j[0]$ lies in (as already mentioned in Definition~\ref{Definition_Quant_Privacy}).
\end{proof}

\begin{remark}
Note here that decomposing the initial state of every $v_j \in \mathcal{V}_p$ into $\mathcal{D}_{max}^+ + 2$ substates is essential for privacy preservation. 
The first set of substates is used as $v_j$'s initial state. 
Then, $v_j$ transmits the rest $\mathcal{D}_{max}^+ + 1$ substates towards its out-neighbors. 
This means that $v_j$ transmits at least one set of privacy variables to every out-neighbor. 
As a result, if $v_{l'} \in \mathcal{V}_p$ (where $v_{l'} \in \mathbb{N}_j^+$), then $v_{l'}$ receives (and sums with its own mass variables) at least one set of $v_j$'s privacy variables before it transmits every set of its own privacy variables towards its out-neighbors. 
\end{remark}

In the above analysis, we are interested in whether the curious nodes can exactly infer the value of another node. 
The case where a set of curious nodes attempts to ``estimate'' the initial values is outside the scope of this paper and will be considered as a future direction.

\section{SIMULATION RESULTS}\label{results}

In this section, we illustrate the behavior of Algorithm~\ref{algorithm_max} and the advantages of its operation. 
%We analyze the cases of: \\
We analyze the scenario of $1000$ randomly generated digraphs of $20$ nodes each where, the initial quantized state of each node remained the same (for each one of the $1000$ randomly generated digraphs); thus, the average of the nodes' initial states remained equal to $q = 13.4$. 

\begin{figure}[t]
\begin{center}
\includegraphics[width=.97\columnwidth]{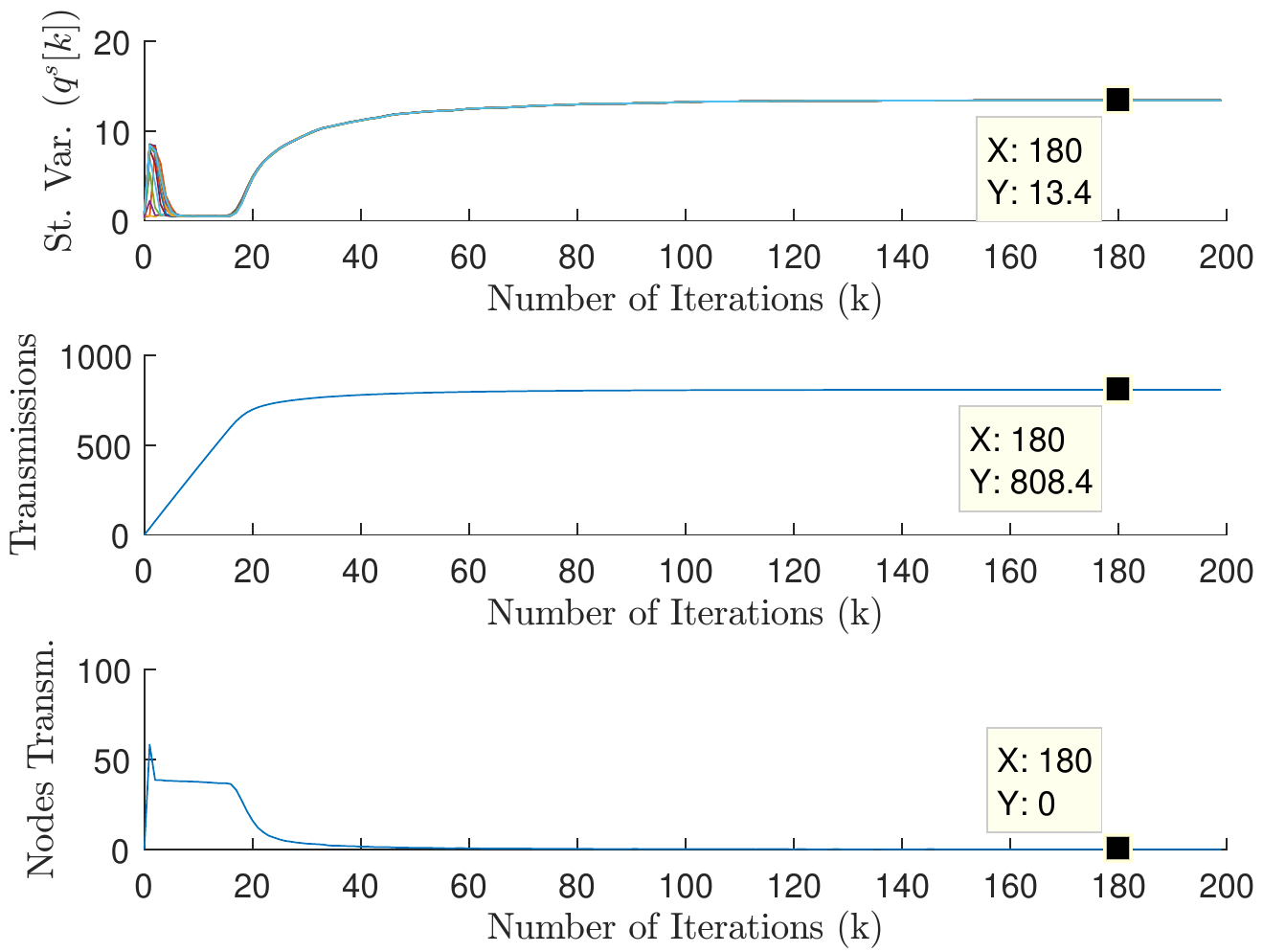}
\caption{Execution of Algorithm~\ref{algorithm_max} averaged over $1000$ random digraphs of $20$ nodes.  
\textit{Top figure:} Average values of node state variables plotted against the number of iterations (averaged over $1000$ random digraphs of $20$ nodes). \textit{Middle Figure:} Average total number of transmissions plotted against the number of iterations (averaged over $1000$ random digraphs of $20$ nodes). \textit{Bottom Figure:} Average number of nodes performing transmissions plotted against the number of iterations (averaged over $1000$ random digraphs of $20$ nodes)\vspace{-0.2cm}.}
\label{20nodes_13_4_average1000_transm}
\end{center}
\end{figure}

In Fig.~\ref{20nodes_13_4_average1000_transm}, we illustrate Algorithm~\ref{algorithm_max} over a random digraph of $20$ nodes where the average of the initial states is $q = 13.4$. 
We show the average number of time steps needed for quantized average consensus to be reached, the average number of transmissions accumulated until each time step, and the average number of nodes performing transmissions at each time step. 
We observe that Algorithm~\ref{algorithm_max} converges after $180$ time steps, with the average total number of transmissions performed until $180$ time step being equal to $808.4$. 
Additionally, we observe that the average number of nodes performing time steps at each iteration becomes almost equal to zero after $50$ time steps, and becomes eventually equal to zero after $180$ time steps.

\begin{figure}[t]
\begin{center}
\includegraphics[width=.95\columnwidth]{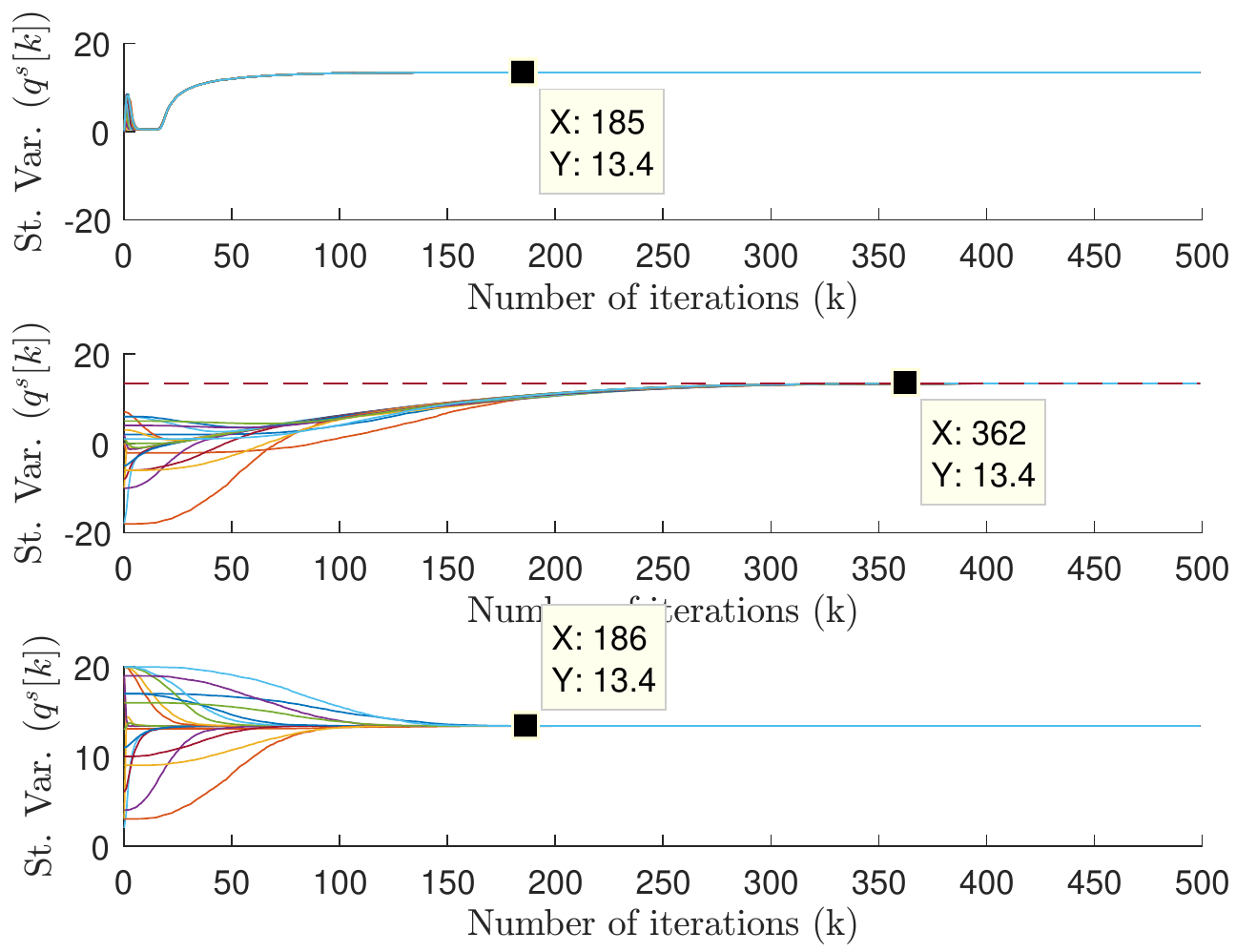}
\caption{Comparison between Algorithm~\ref{algorithm_max}, the event-based offset algorithm in \cite{2020:Rikos_Privacy_CDC}, and the initial zero-sum offset algorithm in \cite{2021:Rikos_Privacy_TCNS}.  
\textit{Top figure:} Average values of node state variables plotted against the number of iterations for Algorithm~\ref{algorithm_max} (averaged over $1000$ random digraphs of $20$ nodes). \textit{Middle Figure:} Average values of node state variables plotted against the number of iterations for the event-based offset algorithm in \cite{2020:Rikos_Privacy_CDC} (averaged over $1000$ random digraphs of $20$ nodes). \textit{Bottom Figure:} Average values of node state variables plotted against the number of iterations for the initial zero-sum offset algorithm in \cite{2021:Rikos_Privacy_TCNS} (averaged over $1000$ random digraphs of $20$ nodes)\vspace{-0.2cm}.}
\label{Comparisons_20nodes_13_4}
\end{center}
\end{figure}

In Fig.~\ref{Comparisons_20nodes_13_4} we plot the node state variables averaged over $1000$ randomly generated digraphs of $20$ nodes each, where the average of the initial states is $q = 13.4$. 
We compare Algorithm~\ref{algorithm_max} against (i) the event-based offset algorithm in \cite{2020:Rikos_Privacy_CDC} (see middle of Fig.~\ref{Comparisons_20nodes_13_4}), and (ii) the initial zero-sum offset algorithm in \cite{2021:Rikos_Privacy_TCNS} (see bottom of Fig.~\ref{Comparisons_20nodes_13_4}). 
In \cite{2020:Rikos_Privacy_CDC} (case (i)), the initial offset for every node $v_j$ is $u_j \in [-100, -50]$ and the offset adding steps are $L_j \in [20, 40]$ during the execution. 
In \cite{2021:Rikos_Privacy_TCNS} (case (ii)), the initial offset for every node $v_j$ is $u_j \in [-100, 100]$ and the offsets are $u^{(l)}_j \in [-20, 20]$, for every $v_l \in \mathcal{N}_j^+$. 
We observe that Algorithm~\ref{algorithm_max} converges after $185$ time steps and again significantly outperforms the event-based offset algorithm in \cite{2020:Rikos_Privacy_CDC} which converges after $362$ time steps. 
Furthermore, it is interesting to note that Algorithm~\ref{algorithm_max} requires almost the same time steps as  the initial zero-sum offset algorithm \cite{2021:Rikos_Privacy_TCNS} which converges after $186$ time steps. 
However, note again that in \cite{2021:Rikos_Privacy_TCNS} each node performs multiple simultaneous transmissions of different quantized values during the Initialization operation, as mentioned in Section~\ref{sec:compar_subsec}. 
Finally, note that neither \cite{2020:Rikos_Privacy_CDC} nor \cite{2021:Rikos_Privacy_TCNS} exhibit finite transmission capabilities. 
This makes Algorithm~\ref{algorithm_max} the first algorithm in the current literature in which each node (i) achieves the exact quantized average of the initial states, (ii) terminates its transmission operation, and (iii) preserves the privacy of its initial state.

\section{CONCLUSIONS}\label{future}

In this paper, we presented a privacy-preserving event-triggered quantized average consensus algorithm. 
The algorithm allows each node in the network to calculate the exact quantized average of the initial states in the form of a quantized fraction without revealing its initial quantized state to other nodes. 
The privacy-preserving strategy takes full advantage of the algorithm's event-based nature and finite transmission capabilities and allows each node to cease transmissions once convergence has been achieved without knowledge of any global parameter (i.e., network diameter). 
%These advantages make the algorithm suitable for energy limited networks since it prolongs the lifetime of nodes. 
%Furthermore, the processing and exchange of quantized messages between nodes facilitate the use of cryptographic primitives in order to encrypt sensitive data. 
We also analyzed the algorithm's finite time convergence and presented an upper bound on the required number of time steps. 
%We presented an upper bound on the number of transmissions each node performs during the operation of the algorithm.
Then, we presented necessary and sufficient topological conditions under which the proposed algorithm allows nodes to preserve their privacy. 
Finally, we demonstrated the performance of our proposed algorithm and compared it against other algorithms in the existing literature. 

In the future, 
%we plan to extend the operation of the proposed algorithm for the scenario where we have time-varying communication topologies, with bounded or unbounded (packet drops) transmission delays over communication links. 
%Furthermore, 
we plan to extend the algorithm's operation to guarantee privacy preservation for the case where curious nodes attempt to ``estimate'' the initial values of other nodes. 

% to achieve finite time consensus with finite transmission guarantees in the presence of malicious (i.e., misbehaving) nodes. 

%  case where a set of curious nodes attempts to ``estimate'' the initial values is outside the scope of this paper and will be considered as a future direction.

\balance

\vspace{-0.4cm}

% ------------------------------------------------------------------------------
% Bibliography
% ------------------------------------------------------------------------------
\bibliographystyle{IEEEtran}
\bibliography{bibliografia_consensus}

\balance

\end{document}